\documentclass[
 aip,
 floatfix,
 amsmath,amssymb,
 reprint,
author-numerical,
]{revtex4-2}

\usepackage{natbib}

\usepackage{dcolumn}
\usepackage{bm}

\usepackage[T1]{fontenc}

\usepackage[utf8]{inputenc}
\usepackage{natbib}
\usepackage{graphicx}

\usepackage{floatflt}
\usepackage{hyperref}
\hypersetup{colorlinks,allcolors=black}

\usepackage[normalem]{ulem}
\usepackage{url}
\usepackage{svg}
\usepackage{capt-of}
\usepackage{soul}

\usepackage{xcolor}
\usepackage{placeins}
\newenvironment{proof}{\paragraph{Proof:}}{\hfill$\square$}
\newtheorem{theorem}{Corollary}

\definecolor{green}{rgb}{0,0.62,0.45}

\usepackage{mathtools}

\begin{document}

\title{Absence of pure voltage instabilities in the third order model of power grid dynamics}

\author{Moritz Thümler}

\affiliation{Chair for Network Dynamics, Institute for Theoretical Physics and Center for Advancing Electronics Dresden (cfaed), Technische Universität Dresden, 01069 Dresden}

\author{Xiaozhu Zhang}
\affiliation{Chair for Network Dynamics, Institute for Theoretical Physics and Center for Advancing Electronics Dresden (cfaed), Technische Universität Dresden, 01069 Dresden}

\author{Marc Timme}
\affiliation{Chair for Network Dynamics, Institute for Theoretical Physics and Center for Advancing Electronics Dresden (cfaed), Technische Universität Dresden, 01069 Dresden}
\affiliation{Lakeside Labs, Lakeside B04b, Klagenfurt, 9020, Austria}

\date{\today}

\begin{abstract} 
Secure operation of electric power grids fundamentally relies on their dynamical stability properties. For the third order model, a paradigmatic model that captures voltage dynamics, 
three routes to instability are established in the literature, a pure rotor angle instability, a pure voltage instability and one instability induced by the interplay of both. Here we demonstrate that one of these routes, the pure voltage instability, is inconsistent with Kirchhoff's nodal law and thus nonphysical. We show that voltage collapse dynamics nevertheless exist in the absence of any voltage instabilities.
\end{abstract}

\maketitle

\begin{quotation}
nonlinear dynamics, network dynamics, power system stability, susceptance and admittance matrix, synchronization
\end{quotation}

\textbf{
Most aspects of our daily life essentially depend on a reliable supply of electrical power, thereby imposing severe challenges for stable operation of power grids that consist of many generators (producers of electric power) and loads (consumers of power) connected with transmission lines. From a perspective of network dynamical systems, these challenges translate to requiring steady states 
that are (asymptotically) stable against sufficiently small dynamical perturbations 
, such that all dynamical variables relax back to their steady synchronous (phase-locked) state with fixed phase differences, constant overall grid frequency as well as fixed voltage amplitudes. In contrast, instabilities may cause growth or fluctuations of phase differences, deviating and changing frequencies and non-constant voltage levels, all undesired in power grid operation.  
For the most basic model class of power system dynamics that covers voltage dynamics, three routes to instabilities have been established in the literature. Here we demonstrate that only two of these three remain in the physically relevant regime, while the third is excluded due to Kirchhoff's nodal law of current conservation. 
}

\section{Introduction}
Electric power supply substantially relies on the stable power grid dynamics. Two classes of system variables are especially important for reliable grid operation: grid frequency and terminal voltage amplitudes \cite{Casavola2007,machowski2011power,SimpsonPorco2016}. Instabilities to fluctuations and collapse of terminal voltages 
have been identified as key contributing factors for large-scale blackouts, for instance, in the northeastern United States (2003) and Athens/Greece (2004) \cite{machowski2011power,SimpsonPorco2016}. The phenomena of voltage collapse and voltage instability in power system models have been extensively studied in the literature,
see, e.g., \cite{Ayasun2004,Kwatny1986,Sharafutdinov2018}. 

Since more than a decade ago, beginning with the derivation of a dynamic network model from the physics of coupled synchronous machines \cite{Filatrella2008} and its collective dynamical phenomena such as phase-locking and synchronization in larger networks \cite{Rohden2012}, the self-organized nonlinear dynamics of entire power grid networks have drawn vast attention among research communities.
The collective dynamics of such systems were studied with respect to global asymptotic stability \cite{Filatrella2008,Rohden2012,Sharafutdinov2018,Schmietendorf2017}, real world statistical properties of fluctuations \cite{Schfer2018,Anvari2016} and induced response dynamics \cite{Zhang2019} up to dynamically induced cascading failures \cite{yang2017small, Schaefer2018}. All of such works have contributed to a conceptual understanding of the stability properties and in particular various types of instabilities in power grid dynamics on the system's level.
In one of the most fundamental dynamic models, a power grid network consists of nodes that are synchronous machines modeling electrical motors or generators. 
A range of models of this class with various degrees of detail have been studied in the literature \cite{machowski2011power, Witthaut2021Nonlinear}. A most commonly studied model consists of coupled swing equations, employing the \textit{second order model} of synchronous machines \cite{Filatrella2008}. Here, the independent variables describing the state of each machine $i$ are given by the deviation of the power angle $\Theta_i(t)$ from an operating point and its time derivative $\dot{\Theta}_i(t)$ quantifying the local deviation from the grid frequency, with a nominal value of $2 \pi \times 50$ Hz in Europe and $2\pi \times 60$ Hz in the US \cite{Owen1997}. Grid frequency constitutes an important quantity for grid operators to control the dynamical state of power grids \cite{Casavola2007,machowski2011power}.
The second order model of synchronous machines takes the terminal voltage amplitudes $E_i$ to be constant and therefore cannot address any instabilities resulting from the dynamics of voltages. The \textit{third order model} constitutes the next higher order model and enables a dynamical description of terminal voltage amplitudes \cite{machowski2011power,Schmietendorf2017,Sharafutdinov2018}. In particular,  three 
routes to instability are established in the literature \cite{Schmietendorf2017, Sharafutdinov2018} for the third order model: one pure rotor angle instability, one pure voltage instability as well as an instability related to the interplay of rotor angle and voltage dynamics. In this work we differentiate between linear (asymptotic) stability of the voltage subsystem, known as the pure voltage instability in the literature, and alterations of voltage variables upon parameter changes that are not related to a change of the linear stability of the voltage subsystem. We refer to the first one as \emph{voltage instability} or instability of the voltage subsystem, and the latter as \emph{voltage collapse}. 

In this article, we demonstrate that the pure voltage instability in the third order model is inconsistent with Kirchhoff's nodal law (also known as Kirchhoff's current law) and thus physically impossible. It emerges as an artifact of extending the parameter regime of the model to nonphysical configurations that violate Kirchhoff's nodal law. Without the constraint of Kirchhoff's nodal law, voltage instability may or may not emerge in the third order model, depending on the choice of system parameters. Employing Gershgorin's circle theorem, we analytically show that the 
relevant eigenvalues of the local Jacobian stay negative and bounded away from zero if Kirchhoff's law is respected. Thus, instabilities of the voltage subsystem are not captured by the third order model in the regime that is physically relevant. Moreover, we numerically demonstrate that voltage collapse is still observable in the third order model within the physically relevant parameter regime if reactive power demand cannot be met due to limitations in the dynamic transmission capacities.

\section{Necessary conditions for pure voltage instabilities in the third order model \label{sec:conditions}}

The loss of acceptable voltage levels has been observed in different forms in real world power systems \cite{Casavola2007}. Mathematical models of power systems predict the existence of both, voltage collapse and instabilities and capture transitions from normal operation to dysfunctional states by bifurcations induced by varying parameters across specific critical values. 

Let us consider the third-order model, a dynamical systems model of a power grid that consists of $N$ generators and consumers modeled as synchronous machines which are interconnected by alternating current (AC) transmission lines. 
The third-order model captures three dynamical variables per node $i$, a phase angle $\Theta_i(t)$, its instantaneous rotation frequency $\omega_i(t)=d\Theta_i/dt(t)$ and a voltage amplitude $E_i(t)$. 
The dynamics of one synchronous machine $i$ reads \cite{machowski2011power,Sharafutdinov2018}
\begin{align}
    \dot{\Theta}_i &= \omega_i \nonumber \\
    \dot{\omega}_i &= P_i - \alpha_i \omega_i - P^{\text{el}}_i(\boldsymbol{\Theta}, \boldsymbol{E}), \nonumber \\
    \dot{E}_i &=  E^f_i - E_i + X_i I_i(\boldsymbol{\Theta}, \boldsymbol{E}). \label{eq:third}
\end{align}
where the dot denotes differentiation with respect to time $t$. Here $\boldsymbol{\Theta} \in \mathbb{R}^N$ denotes the vector of the power angles, $\boldsymbol{\omega} = \boldsymbol{\dot{\Theta}}$ the angular frequency, both with respect to the grid reference frame (rotating at e.g., $\Omega = 2\pi \times 50$ Hz in Europe), and $\boldsymbol{E} \in \mathbb{R}_+^N$ the vector of terminal voltage amplitudes. Here, $\mathbb{R}_+$ denotes the set of non-negative real numbers such that each component $E_i\geq 0$. The remaining machine parameters are the power input or output $P_i\in\mathbb{R}$ (negative for consumers and positive for generators), the mechanical damping $\alpha_i >0$, the voltage set point $E^f_i > 0$, and the reactance $X_i\geq 0$ of the synchronous machine $i$. The coupling functions  $P_i^{\text{el}}:\,\mathbb{R}^{N} \times \mathbb{R}^N \rightarrow \mathbb{R}$ and $I_i:\,\mathbb{R}^{N} \times \mathbb{R}^N \rightarrow \mathbb{R}$ represent respectively the electrical powers and the currents exchanged between the $N$ synchronous machines through the transmission lines. 

An alternating current (AC) transmission line $(i,j)$ between adjacent nodes $i$ and $j$ is modelled by its admittance $Y_{ij}=G_{ij}+\imath B_{ij}\in \mathbb{C}$, where $G_{ij}$ is the conductance and $B_{ij}$ the susceptance. In general, $B_{ij}$ can be positive or negative. However, Ohm's law for amplitudes, such as $E_i$, involves the absolute value of the admittances 
\begin{align}
|Y_{ij}|= \sqrt{G_{ij}^2 + B_{ij}^2}.
\end{align} 
For our application it is therefore sufficient to consider the absolute values of the susceptances  $|B_{ij}|\geq0$ and for simplicity we just assume $B_{ij}\geq0$ for $i\neq j$. Lossless transmission, i.e.,~neglecting Ohmic losses $(G_{ij}=0)$, is taken as a reasonable assumption in high voltage power grid modelling \cite{machowski2011power}, such that the exchanged powers and currents read \cite{machowski2011power,Schmietendorf2017}
\begin{align}
    P^{\text{el}}_i(\boldsymbol{\Theta}, \boldsymbol{E}) &= \sum_{j=1}^N  B_{ij} E_i E_j  \sin(\Theta_i - \Theta_j)  \nonumber \\
    I_i(\boldsymbol{\Theta},\boldsymbol{E}) &= \sum_{j=1}^N B_{ij}  E_j \cos(\Theta_j - \Theta_i). \label{eq:coup}
\end{align}
Here symmetric susceptances ($B_{ij}=B_{ji}$) constitute a symmetric susceptance matrix $B=B^\mathsf{T} \in \mathbb{R}^{N\times N}$, with the diagonal elements $B_{ii}<0$ being self-susceptances. 
For the moment we describe the self-susceptances $B_{ii}$ in relation to the off-diagonal elements according to
\begin{align}
    B_{ii} = \Gamma_i - \sum_{j=1, j\neq i}^N B_{ij} < 0 \label{eq:paramter}
\end{align}
with an additional shunt susceptance parameter $\Gamma_i \in \mathbb{R}$. First, we consider the parameters $\Gamma_i$ as free model parameters and study their implications on the system's stability. Later we will discuss the physically relevant choice of $\Gamma_i$. The system Eq.\,\eqref{eq:third} with substituted coupling functions Eq.\,\eqref{eq:coup} reads
\begin{subequations}
\begin{align}
    \dot{\Theta}_i =& \omega_i \label{eq:third_1} \\
    \dot{\omega}_i =& P_i - \alpha_i \omega_i + \sum_{j=1}^N B_{ij} E_i E_j \sin(\Theta_j - \Theta_i) \label{eq:third_2} \\
    \dot{E}_i =& E^f_i + (X_i \Gamma_i -1) E_i \nonumber  \\&+  X_i\sum_{j=1,j\neq i}^N B_{ij} ( E_j\cos(\Theta_j - \Theta_i) - E_i) . \label{eq:third_3}
\end{align}
\end{subequations}
Power grids are operated near an equilibrium state which for the third-order model is a fixed point 
\begin{align}
\left(\boldsymbol{\Theta^*},\ \boldsymbol{\omega^*},\ \boldsymbol{E^*}\right), \label{eq:equilibrium}
\end{align}
given by a simultaneous solution to Eq.\,\eqref{eq:third_1}-\eqref{eq:third_3} at zero rates of change,
\begin{align} 
\dot{\Theta}_i =  \dot{\omega}_i =  \dot{E}_i = 0 \text{  for all  } i \in \{1,2,...,N\}. \label{eq:velocities}
\end{align}
The existence of equilibria depends on the specific choices of the nodal parameters $E^f_i, P_i, X_i$, the line susceptances $B_{ij}$ and $\Gamma_i$. For instance, an equilibrium only exist if the powers $P_i$ are in balance
\begin{align}
    0 = \sum_{i=1}^N P_i.
\end{align}
Furthermore, from the paradigmatic Kuramoto model\cite{Strogatz2000} it is well known that the coupling strengths have to be sufficiently large to compensate the powers $P_i$ in order to allow the system to settle into an equilibrium. Since for the third order model \eqref{eq:third_3} the equilibrium coupling strengths 
\begin{align}
    K_{ij} \coloneqq B_{ij}E_i^* E_j^*
\end{align}
are bound by 
\begin{align}
    K_{ij} \leq B_{ij} (E^f_l + X_l \mu)^2 \leq B_{ij} (E_l^f)^2,
\end{align}
with the index $l$ denoting the largest equilibrium voltage amplitude $E_l^* \geq E_i^*$ for all $i\in\{1,2,...,N\}$ and $\mu\leq 0$, we conclude that $E^f_i$ and the susceptances $B_{ij}$ with $i\neq j$ have to be sufficiently large. Furthermore, the reactances $X_i$ have to be sufficiently small. We derive necessary conditions for the existence of a fixed point in the appendix. 

Fixed voltages $E_i^*$ and fixed frequencies $\omega_i^*$ are desired in power grid operations, as well as that the system relaxes back to equilibrium when exposed to small perturbations.

Whether the system relaxes back towards the equilibrium, is characterized by the linear stability of the corresponding fixed point of the system\cite{strogatz:2000}. At a fixed point $\left(\boldsymbol{\Theta^*},\ \boldsymbol{0},\ \boldsymbol{E^*}\right)$, the evolution of the linear response $\left(\boldsymbol{\vartheta}, \boldsymbol{\nu}, \boldsymbol{\epsilon}\right)$ of the system [\eqref{eq:third_1}-\eqref{eq:third_3}] is governed by
%
% \begin{align}
% 	 \Theta_i(t) &= \Theta^*_i + \vartheta(t) \nonumber \\
%     \omega_i(t) &= 0 +  \nu_i(t) \nonumber \\
%     E_i(t) &= E^*_i + \epsilon_i(t)
% \end{align}
%
\begin{align}
    \begin{pmatrix} \boldsymbol{\dot{\vartheta}} \\
    \boldsymbol{\dot{\nu}} \\
    \boldsymbol{\dot{\epsilon}}
    \end{pmatrix} = \begin{pmatrix}
                0 & I_N & 0 \\
                 \Lambda & -\alpha I_N &  A  \\
                 A^\mathsf{T} & 0 &  C \end{pmatrix}  \begin{pmatrix} \boldsymbol{{\vartheta}} \\
    \boldsymbol{{\nu}} \\
    \boldsymbol{{\epsilon}} 
    \end{pmatrix} \eqqcolon J  \begin{pmatrix} \boldsymbol{{\vartheta}} \\
    \boldsymbol{{\nu}} \\
    \boldsymbol{{\epsilon}} 
    \end{pmatrix} , 
\end{align}
where $I_N \in \mathbb{R}^{N\times N}$ denotes an identity matrix and $\Lambda, A, C\in \mathbb{R}^{N \times N}$ are submatrices of the Jacobian matrix $\boldsymbol{J}$. The submatrices are defined via their matrix elements 
\begin{eqnarray}
\Lambda_{ij} &=& \left\{ \begin{array}{l l} B_{ij}  E_i^*E_j^*    \cos(\Theta_j^*-\Theta_i^*)   &\text{ for } i \ne j\\ -\sum_{k\ne i} B_{ik}  E_i^* E_k^*  \cos(\Theta_k^*-\Theta_i^*)   &\text{ for } i = j \end{array} \right. \nonumber \\
A_{ij} &=& \left\{ \begin{array}{l l} B_{ij} E_i^*  \sin(\Theta_j^*-\Theta_i^*) &\text{      \ \ \ \ \ \ \ \ \ \ for } i \ne j \\ \sum_k B_{ik} E_k^*  \sin(\Theta_k^* -\Theta_i^*) &\text{ \ \ \ \ \ \ \ \ \ \ for } i = j \end{array} \right. \nonumber \\
C_{ij} &=&  \left\{ \begin{array}{l l} X_i B_{ij} \cos(\Theta_j^* - \Theta_i^*) &\text{ \ \ \ \ \ \ \ \ \ \ for } i \ne j \\ X_i\Gamma_i  - 1 -X_i\sum_{k\neq i} B_{ik}  &\text{ \ \ \ \ \ \ \ \ \ \ for } i = j\end{array}   \right.. \nonumber \\ \label{eq:jac_def}
\end{eqnarray}
The matrix $\boldsymbol{J}$ has one eigenvalue $\lambda_0=0$
%\outgreen{, see references\cite{Rohden2012,Sharafutdinov2018,Schmietendorf2017} belonging to}
corresponding to the eigenvector $\boldsymbol{v}_0 = (\boldsymbol{1},\boldsymbol{0}, \boldsymbol{0})^\mathsf{T}$, indicating that the system is marginally stable along $\boldsymbol{v}_0$ \cite{Rohden2012,Sharafutdinov2018,Schmietendorf2017}. 
Nevertheless, since 
a shift along $\boldsymbol{v}_0$ does not change the physical state of the system, we %\outgreen{. We}
thus only consider the system's linear stability in the orthogonal space\cite{Sharafutdinov2018}
\begin{align}
    \mathcal{D}^\perp = \left\{ \boldsymbol{x} \in \mathbb{R}^{3N}| \boldsymbol{x} \boldsymbol{v_0} = 0 \right\}.
\end{align}
As shown by Sharafutdinov et al. \cite{Sharafutdinov2018}, the asymptotic stability of the system in 
$\mathcal{D}^\perp$ (a negative definite $J$) implies that both submatrices $\Lambda$, the rotor angle subsystem, and $C$, the voltage subsystem, are negative definite themselves, i.e.,
\begin{align}
    J \text{ is negative definite} \Rightarrow \Lambda \text{ and } C \text{ are negative definite.}
\end{align}
In this way, three routes to instability in the third order model of synchronous machines are established \cite{Sharafutdinov2018}: One pure rotor angle instability, where $\Lambda$ loses negative definiteness; one pure voltage instability, where $C$ loses negative definiteness; and a third route resulting from an interplay between both subsystems where a fixed point for both voltage and rotor angle equation cannot be determined simultaneously.

In particular, if the real parts of any eigenvalue of either one of the two submatrices $C$ or $\Lambda$ crosses zero from below (excluding $\lambda_0 = 0$ for $\Lambda$), the entire systems' equilibrium becomes linearly unstable. Related earlier work has shown that one condition 
for $\Lambda$ to be negative definite is\cite{machowski2011power,Manik2014}
\begin{align}
    |\Theta_j - \Theta_i| \leq \frac{\pi}{2} \label{eq:stab_rot}
\end{align}
for all adjacent synchronous machines $i$ and $j$, i.e., those directly connected by a transmission line. We now focus on the analysis of the voltage subsystem characterized by the matrix $C$ by applying the Gershgorin disk theorem \cite{gerschgorin31}. 
The broadly applicable theorem states \cite{gerschgorin31,stoer1989numerische, timme2004topological, timme2008simplest, freund2007stoer} that for any square matrix $M \in\mathbb{C}^{N\times N}$ all the eigenvalues $\lambda_j^{M}$ for all $j \in \{1,2,...,N\}$ are in the union 
\begin{align}
    \lambda_j^M \in \bigcup\limits_{i=1}^{N} \mathcal{G}_i ,
\end{align}
of $N$ disks 
%$\mathcal{G}_i$ 
%
\begin{align}
    \mathcal{G}_i \coloneqq  \left\{ z \in \mathbb{C}\ \ |\ \  |z - M_{ii}| \leq \sum_{j\neq i} |M_{ij}| \right\}.
\end{align}
The diagonal elements $M_{ii}$ define the center of the disk, while the sum across the absolute values of the off-diagonal elements of the same row defines its radius. 
Since linear stability of the voltage subsystem alone is ensured if all eigenvalues $\lambda_i^C$ of the matrix $C$ have a negative real part, we evaluate under which conditions all the Gershgorin disks are entirely on the left-hand side of the imaginary axis. 

To this end, we define the directed margin $d_i$

\begin{align}
    d_i \coloneqq \sup\left\{ \text{Re}(q)\, |\, q\in \mathcal{G}_i \right\},
\end{align}
between the imaginary axis and the Gershgorin disk. A negative margin for all $i$ ensures linear stability of the voltage subsystem characterized by $C$. Thus, for parameters where all $d_i< 0$, voltage instabilities do not occur.
\begin{figure}[!ht]
    \centering
    \includegraphics{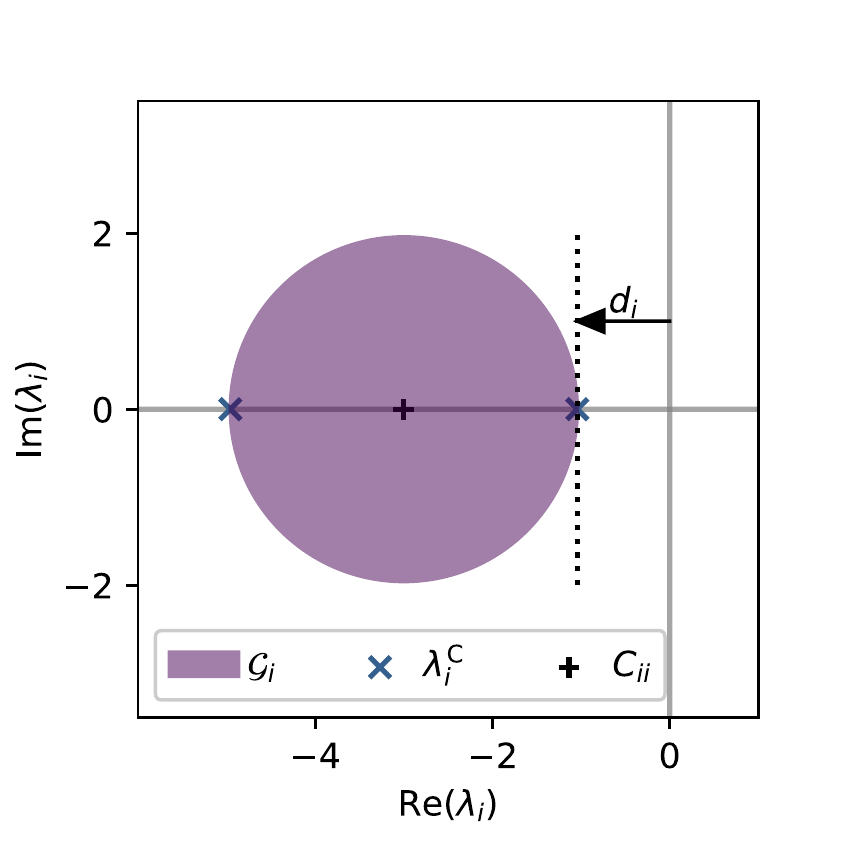}
    \caption{\textbf{Negative margins $d_i$ ensure linear stability of the voltage subsystem:} The voltage subsystem of $N=2$ coupled third-order synchronous machines is evaluated in terms of the margin $d_i$. The two Gershgorin disks coincide, and the true positions of the eigenvalues are precisely on the borders of the disks. The parameters are $P_1=-P_2 = 1.5$, $B_{12}=B_{21}=2$, $E^f=2$, $1.0$ and $\Gamma_1 = \Gamma_2 = 0$.
    }
    \label{fig:margin}
\end{figure}
For general setups of the network and machine parameters, the margins $d_i$ of the symmetric matrix $C$ satisfy

\begin{align}
    d_i =& C_{ii} + \sum_{j\neq i }^N |C_{ij}| \nonumber \\
      =&  -1+X_i \Gamma_i -X_i \sum_{j=1,j\neq i}^N B_{ij} \nonumber \\ &+ X_i \sum_{j=1,j\neq i}^N |B_{ij} \cos(\Theta_j^* - \Theta_i^*)| \nonumber \\
      =& -1+X_i \Gamma_i+X_i \sum_{j\neq i}^N B_{ij} (|\cos(\Theta_j^* - \Theta_i^*)| - 1) \nonumber \\
      \leq& -1 +  X_i \Gamma_i . \label{eq:margin}
\end{align}
In the first step, we apply the definition of the Gershgorin disk $\mathcal{G}_i$ to the matrix $M=C$. In the second step, we substitute the matrix elements of $C$ according to Eq.\,\eqref{eq:jac_def}, exploiting that $X_i>0$. In the third step, as $B_{ij}>0$ for $i\neq j$, we factor it out and regroup the terms. Finally, bounding the cosine function by its upper bound $1 = \max \{ \cos(x)\, | \, x\in\mathbb{R}\}$ provides an upper bound for $d_i$. 
We set the upper bound $ X_i \Gamma_i - 1  = 0$  of $d_i$ and obtain
\begin{align}
    X_{\text{crit}}(\Gamma) = \frac{1}{\Gamma},%\green{\text{[@Moritz: ``='' here?]}}
\end{align}
a lower bound for the critical parameter $X_{\text{crit}}$ with $\Gamma = \max  \{\Gamma_i\,|\, i \in \{1,2,...,N\} \} $. For all $0 \leq X \leq X_{\text{crit}} (\Gamma)$ the matrix $C$ is negative definite as shown via the Gershgorin disk theorem. For $X>X_{\text{crit}}$ the matrix may have positive eigenvalues $\lambda_i^{C}>0$ but due to the upper bound approximation in Eq.\,\eqref{eq:margin} this is not guaranteed, hence referred to as potentially unstable region in Fig.\,\ref{fig:parameterspace*}. For our further analysis, we will rely on the stable regime and do not need further knowledge about the potentially unstable region.
We have derived a bound $X_\text{crit}(\Gamma)$ for, which
\begin{align}
	&X< X_{\text{crit}}(\Gamma)\iff\text{voltage stability} \nonumber \\
	&X\geq X_{\text{crit}}(\Gamma)\iff\text{potential voltage instability},
\end{align}
holds.

\begin{figure}[!ht]
    \centering
    \includegraphics{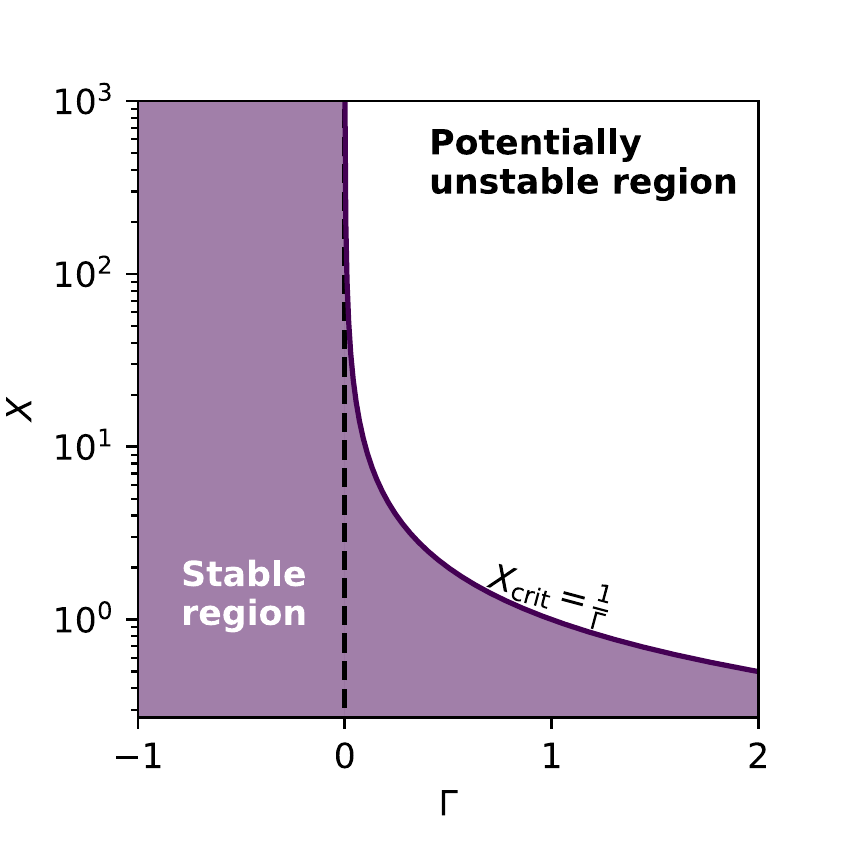}
    \caption{\textbf{The minimal stable region of the voltage subsystem depends on the parameter $\Gamma$:} The matrix of the voltage subsystem  is negative definite in the purple region, such that no pure voltage instability occurs here. In particular, for $\Gamma \leq 0$ the voltage subsystem alone is linearly stable for all $X\in \mathbb{R}^+$. This analysis holds for arbitrary system sizes $N$, parameter settings and thus network topologies.}
    \label{fig:parameterspace*}
\end{figure}
\section{Absence of pure voltage instabilities}
The above analysis proves for $X_i \leq X_\text{crit}(\Gamma)$ the linear stability of the voltage subsystem. However, it does not take into account whether fixed points exist in the potentially unstable region at all. Hence, we do not know at this point whether a transition to an unstable voltage subsystem at all is possible or not. 
For instance, in the simplest system of $N=2$ coupled third-order synchronous machines, one can show that all fixed points are found in the stable region of the voltage subsystem for arbitrary choices of $\Gamma$. The equilibrium of this system configuration is explicitly given via the set of equations,
\begin{align}
    0 &= \omega_1^* \nonumber \\
    0 &= \omega_2^* \nonumber \\
    0 &= P_1 - \alpha_1 \omega_1^* + B_{12} E_1^* E_2^* \sin(\Theta_2^* - \Theta_1^*) \nonumber \\
    0 &= P_2 - \alpha_2 \omega_2^* + B_{12} E_1^* E_2^* \sin(\Theta_1^* - \Theta_2^*) \nonumber \\
    0 &= E^f - E_1^*  + X\Gamma E_1^*  + X B_{12} ( E_2^* \cos(\Theta_2^* - \Theta_1^*) - E_1^*  ) \nonumber \\
    0 &= E^f - E_2^*  + X\Gamma E_2^*  + X B_{12} ( E_1^* \cos(\Theta_2^* - \Theta_1) - E_2^* ), 
\end{align}
which effectively reduces to
\begin{subequations}
\begin{align}
    0 &= P + B_{12} E^* E^* \sin(\Delta \Theta^*) \label{eq:fp2_1} \\
    0 &= E^f - E^* + X\Gamma E^* + X B_{12} ( E^* \cos(\Delta \Theta^*) - E^*  ) \label{eq:fp2_2},
\end{align}
\end{subequations}
with $E_1 = E_2 = E^*$ (which follows from subtracting the voltage equations from one another) and $\Delta \Theta^* = \Theta_2^* - \Theta_1^*$. In this configuration, $C$ reads
\begin{align}
    C = \begin{pmatrix}   
                -X(B_{12}-\Gamma) - 1 & XB_{12} \cos(\Delta \Theta) \\
                XB_{12}\cos(\Delta \Theta) & - X(B_{12}-\Gamma) - 1
    \end{pmatrix}
\end{align}  
and has the two real eigenvalues 
\begin{align}
    \lambda_{\pm}^{C} =  -X(B_{12}-\Gamma) - 1 \pm   XB_{12} \cos(\Delta \Theta).
\end{align}
To investigate where the voltage subsystem changes its linear stability, we analyze the point where its largest eigenvalue $\lambda_+^{C} = 0$ such that
\begin{align}
      XB_{12} \cos(\Delta \Theta) &= X(B_{12}-\Gamma) + 1.
\end{align}
We substitute the latter into Eq.\,\eqref{eq:fp2_2} and find
\begin{align}
    0 &= E^f - E^* + X\Gamma E^* + E^* (  XB_{12} - X\Gamma + 1 )  - X B_{12}E^*  \nonumber \\
    0 &= E^f. \nonumber \\
\end{align}
Given that $E^f$ is a strictly positive machine parameter, we conclude that for $N=2$, no transition from a stable equilibrium towards an unstable equilibrium of the voltage subsystem exists, under any configuration of all parameters  of the model system. This finding constitutes a contradiction to previously published statements in the literature about an instability of the $2$-node third-order model of synchronous machines\cite{Sharafutdinov2018}, where for a positive $\Gamma$ a transition towards an unstable voltage subsystem has been demonstrated numerically.  
The estimator $X_{\text{crit}}(\Gamma)$ shows no options for pure voltage instabilities
for general system configurations if $\Gamma_i$ for all $i$ are set to negative values. But what are physically possible choices for $\Gamma$?

\section{Kirchhoff's nodal law requires $\Gamma<0$}

In this section we motivate the physically relevant parameter ranges for the shunt susceptances $\Gamma$ starting with a sample grid and deriving the form of the nodal susceptance matrix based on Kirchhoff's nodal law (also known as Kirchhoff's current law).
Let us consider a network of electric transmission lines as depicted in Fig\,\ref{fig:sample_grid}. The nodes $i\in \{1,2,3\}$, simply intersections of transmission lines, connect the synchronous machines (as indicated by the dashed boxes) with the transmission lines of the power grid. 
Kirchhoff's nodal law states that the currents flowing into any such intersection nodes have to balance the current flowing out of the node. 
In other words, charge is a conserved quantity. Kichhoff's nodal law is closely interlinked with a specific choice of the shunt susceptance parameter $\Gamma$. 
For illustration, we apply the current law to the network configuration displayed in Fig.\,\ref{fig:sample_grid}. We know, for instance, that the sum across all currents at the node $i=1$ has to be zero
\begin{align}
    0 &= \tilde{I}_1 + \tilde{I}_{21} + \tilde{I}_{31} + \tilde{I}_{1a} + \tilde{I}_{1b} \label{eq:kirch},
\end{align}
with the current $\tilde{I}_1$ to be the one flowing into the synchronous machine $1$ indicated by the left dashed box in Fig.\,\ref{fig:sample_grid}. The currents $\tilde{I}_{xy}$ refer to the currents flowing across the susceptance $B_{xy}$.
The tilde $\tilde{\cdot}$ indicates that the quantity is described in the grid reference frame rather than in the reference frame of the synchronous machine. Essentially, voltages on two perpendicular rotor axes (of which one is assumed to be identically zero in the third order model) of the synchronous machine are rotated by the power angle $\Theta_i$ to the grid reference frame. More details are provided in, e.g., Schmietendorf et al.\cite{Schmietendorf2017}.
In order to compute the current, $\tilde{I}_1$ we resolve Eq.\,\eqref{eq:kirch} and obtain
\begin{align}
    \tilde{I}_1 =  - \tilde{I}_{21} - \tilde{I}_{31} - \tilde{I}_{1a} - \tilde{I}_{1b} \label{eq:kirch2}.
\end{align}

Furthermore, applying Ohm's law $I = B \Delta E$ to both lines, we find
\begin{align}
    \tilde{I}_1 =  B_{12} (\tilde{E}_2 - \tilde{E}_1 ) + B_{13} (\tilde{E}_3 - \tilde{E}_1)  - B_{1a} \tilde{E}_1  - B_{1b} \tilde{E}_1, \label{eq:shunt}
\end{align}
where we set the ground node's voltage amplitude to zero, $E_0 = 0$. We can write Eq.\,\eqref{eq:shunt} in vector form and find
\begin{align}
    \tilde{I}_1 = \begin{pmatrix} -B_{12}-B_{13} -B_{1a} - B_{1b} & B_{12} & B_{13}  \end{pmatrix} \begin{pmatrix} 
            \tilde{E}_1 \\ \tilde{E}_2 \\ \tilde{E}_3
    \end{pmatrix}
\end{align}
This is the first component of the generalized Ohm's law for the network in Fig.\,\ref{fig:sample_grid} that in general reads 
\begin{align}
    \boldsymbol{\tilde{I}} = B \boldsymbol{\tilde{E}},
\end{align}
with the susceptance matrix $B\in \mathbb{R}^{3\times 3}$, from which we identify the first diagonal element 
\begin{align}
    B_{11} = - B_{12} - B_{13} - B_{1\alpha} - B_{1\beta}.
\end{align}
We thereby identify the parameter $\Gamma_1$ for this system
\begin{align}
\Gamma_1 = -B_{1b} - B_{1a} \leq 0.
\end{align}
This generalizes to the overall form Eq.\,\eqref{eq:paramter} for the self susceptance
\begin{align}
    B_{ii} =  \Gamma_i - \sum_{j\neq i } B_{ij},
\end{align}
for arbitrary network topologies  with the shunt susceptance $\Gamma_i  \leq 0 $ corresponding to the negative of the line susceptances connecting the node to the ground node. According to the upper bound of $d_i$ in Eq.\,\eqref{eq:margin} visualized in Fig.\,\ref{fig:parameterspace*}, a negative shunt susceptance $\Gamma$ automatically ensures linear stability of the voltage subsystem.
We thus showed that the voltage subsystem at an equilibrium is \textbf{always} asymptotically stable if the parameters, in particular the shunt susceptance $\Gamma_i$, reflect Kirchhoff's nodal law.

\begin{figure}[!ht]
    \centering
    \includegraphics{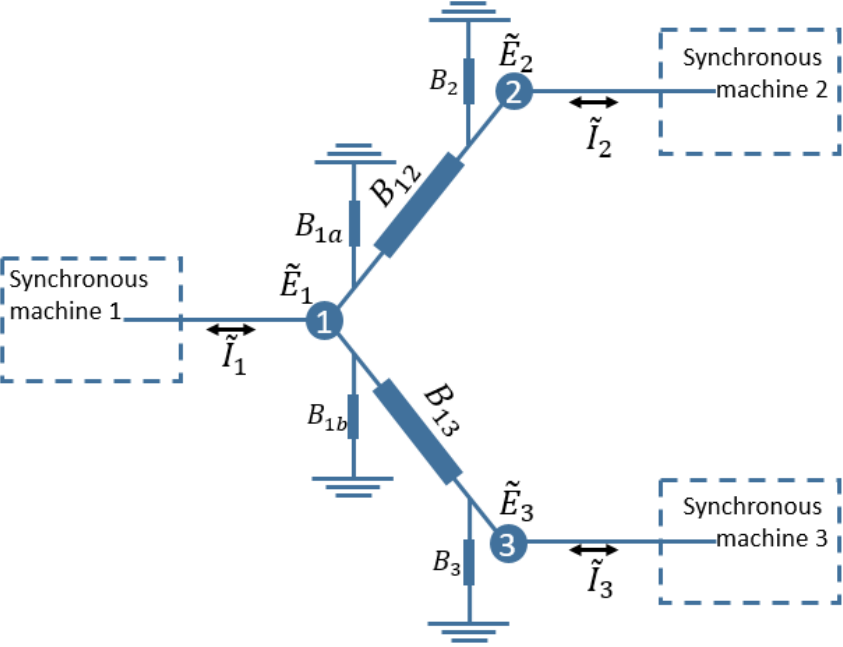}
    \caption{\textbf{Sample grid to explain the properties of the nodal susceptance matrix $\boldsymbol{B}$:} The sample grid consists of three interconnected machines. The variables $\tilde{E}_i$  stand for the voltage at the respective node. The variables $\tilde{I}_i$ represent the current that flows from each node to the attached dashed box.}
    \label{fig:sample_grid}
\end{figure}

By setting the shunt susceptance $\Gamma_i>0$ in recent literature\cite{Sharafutdinov2018,Schmietendorf2017} linearly unstable voltage subsystems were uncovered, which have to be considered nonphysical as violating Kirchhoff's nodal law. 

\section{Voltage collapse in third order synchronous machine dynamics.}
\begin{figure}[!h]
    \centering
    \includegraphics{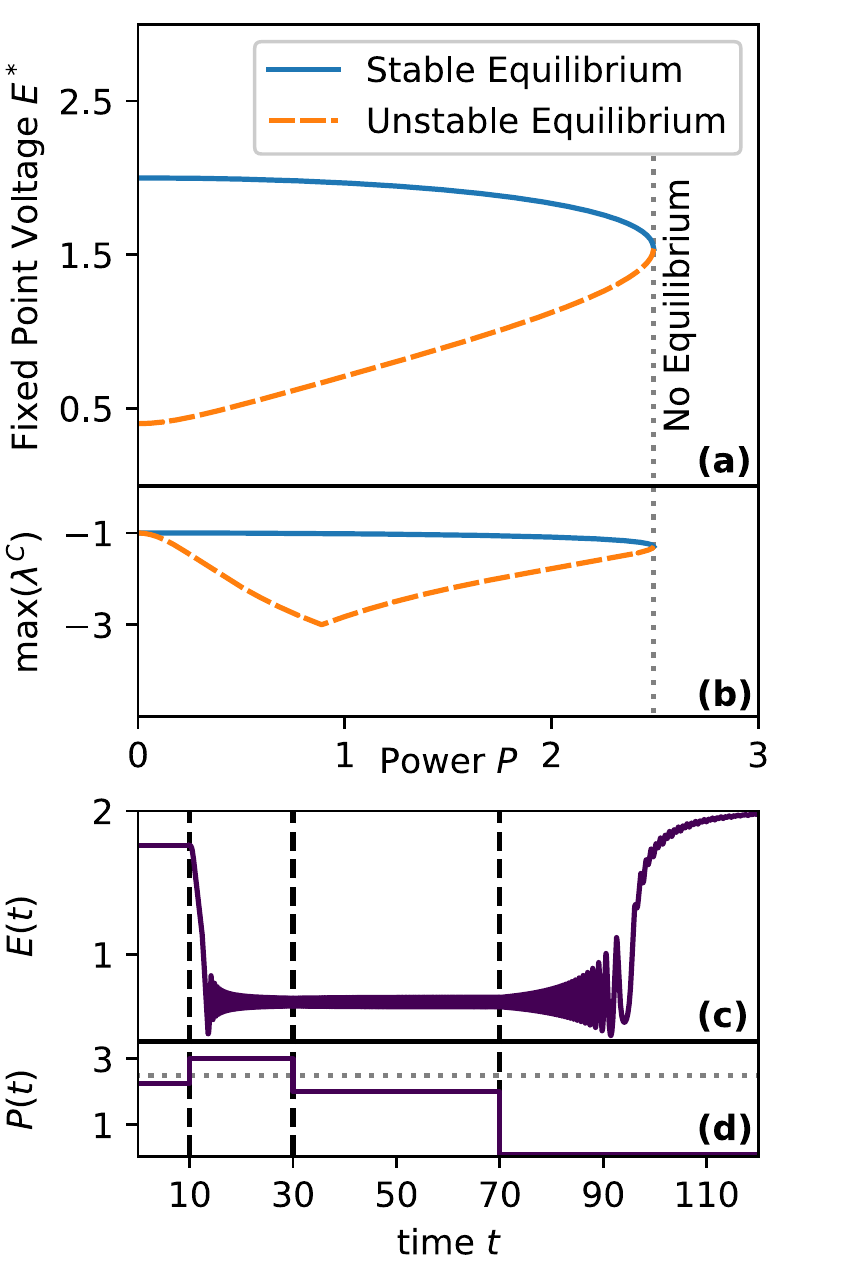}
    \caption{\textbf{Voltage collapse without pure voltage instability.} (a) Equilibrium voltage amplitude $E^*$ changes with the reactive power $P$ in a $N=2$ node system. The system exhibits two relatively large voltage equilibria, one stable (solid line) and one unstable (dashed line) that annihilate in a saddle-node bifurcation (near $P=2.5$), 
    beyond which no real solution for Eq.\,(\ref{eq:fp2_1},\ref{eq:fp2_2}) exists.
    (b) Consistently negative eigenvalues indicate that the voltage subsystem is linearly stable for the whole range of $P$ and the stable and unstable branch of the entire third order dynamics where the fixed point exists.
    (c,d) Time evolution of the voltage amplitude $E(t)$ upon temporally changing power $P(t)$. At $t=10$, the power demand $P$ increases (instantaneously) from $2.25$ to $3$, overloading the line in the current equilibrium. As a consequence, the voltage amplitudes of the system drop substantially, a phenomenon known as voltage collapse \cite{machowski2011power,SimpsonPorco2016}. Even after lowering $P$ again all the way to $2$ at $t=30$, the voltage amplitudes remain low. Only sufficiently small $P$ allow the system to relax back to the stable equilibrium as shown for $t>70$ where the power $P=0.1$. 
    Network parameters are $\Gamma_1 = \Gamma_2 = 0$, $B_{12} = B_{21} =  2$, $X=1.0$, $\alpha = 0.1 $ and $E^f = 2$. 
     }
    \label{fig:voltage_drop}
\end{figure}
Despite the fact that the physical third-order model of synchronous machines does not exhibit pure voltage instabilities, i.e., linearly unstable voltage subsystems, we emphasize that it still captures the phenomenon of voltage collapse, i.e., substantial voltage changes upon parameter changes.
Voltage collapse has been discussed as one of the root causes of various real world power outages\cite{machowski2011power,SimpsonPorco2016}. In this section, we illustrate numerically that the third order model of synchronous machines has the capability to undergo voltage collapse. The underlying cause, instead of a linear instability of the voltage subsystem, is a saddle-node bifurcation at which the existence of two equilibria is lost, including the stable one. The saddle-node bifurcation occurs when transmission line capacities at the respective voltage levels are not sufficient to meet the power demand $P$ of the consumer. We investigate this (see Figure~\ref{fig:voltage_drop}) for a simple system of $N=2$ nodes and one transmission line, as in section \ref{sec:conditions}.

Figure~\ref{fig:voltage_drop} displays the loss of existence of two equilibria upon parameter changes of the reactive power $P$ and a possible way of restoring higher voltage levels. Beyond the critical value of the power $P_{\text{max}} = B_{12} (E^*)^2$, see Eq.\,\eqref{eq:fp2_1} where $P_{\text{max}}$ has to compensate for the power $P$ that needs to be transported across a transmission line, the stable equilibrium is lost and the third order dynamics causes the voltage amplitudes to drop significantly. The third order model thus captures the phenomenon of voltage collapse. However, the root cause is not the loss of the stability of the voltage subsystem but a power overload of the transmission line and the related loss of equilibria. Even at $P$ below the previously valid critical value of the power $P_{\text{max}}$, the system does \textit{not} relax back to the stable equilibrium. Significantly smaller power values $P < B_{12} (E(t))^2 $ are needed to stabilize the power transmission. See Figure \ref{fig:voltage_drop} for details of the example.

\section{Conclusion}
In this article we have analyzed the possibility of pure voltage instabilities  in the third order model and explained why, due to violating Kirchhoff's nodal law, they cannot occur at physically consistent parameters. These findings contrast previously published work \cite{Sharafutdinov2018,Schmietendorf2017} that, however, correctly derived necessary conditions for and numerically exemplified linear stability. 
Interestingly, real world power outages have indeed been tied to effects described as voltage instabilities. However, this terminology referred to voltage drops\cite{Casavola2007,machowski2011power}, which we have observed numerically in the third order model upon changes of parameters without changes in stability of any operating state. The power overload of transmission lines is the root cause for the voltage collapse. We thus emphasize that the term voltage collapse is to be  carefully separated from the term voltage instability, which relies on the linear stability of the voltage subsystem. These two phenomena are  mathematically not connected.
Another class of power system models, given by algebraic differential equations, were studied extensively in the literature\cite{Ayasun2004,Kwatny1986} in terms of voltage collapse. The fundamental difference of that model class is that consumers are assumed to have fixed power angles, $\Theta$ as well as fixed reactive and active power demand. Thus, they represent algebraic constraints to the dynamics of the generators. In such a setup, linearly stable, low voltage equilibria may be identified. It is particularly difficult to operate a system that is trapped at such an equilibrium, and bring it back to a high voltage fixed state\cite{Ayasun2004,Kwatny1986}. A detailed analysis of a three bus system is given in\cite{Beardmore2000}. In contrast to the third order model of power grids, these extended models exhibit changes of local stability properties upon parameter changes.

For the third order model, it is sufficient to ensure that line capacity constraints are satisfied to ensure stable, high voltage operation. Given the results presented above, two research paths open up to further study voltage stability properties in power system models. First, one could factor in Ohmic losses, i.e., $G_{ij}>0$ and analyze whether local stability properties of the voltage subsystem undergo a bifurcation. Second, one could investigate non-local stability properties in the third order model by numerically analyzing basin stability\cite{menck2014} for voltage and rotor angle perturbations. The basin size may depend strongly on the line load. Furthermore, it would be of interest to extend the basin stability argument to the extended model of differential algebraic equations\cite{Ayasun2004,Kwatny1986}.

\section{Acknowledgements}
We thank Malte Schröder and Philip Marszal for valuable comments on the manuscript. We gratefully acknowledge support the Bundesministerium für Bildung und Forschung (BMBF, Federal Ministry of Education and Research) under Grant. No. 03EK3055F.

\section{Conflict of Interest}
The authors have no conflicts to disclose.

\section*{References}
\nocite{*}
\bibliographystyle{unsrtnat}
\bibliography{references}

\begin{thebibliography}{30}
\providecommand{\natexlab}[1]{#1}
\providecommand{\url}[1]{\texttt{#1}}
\expandafter\ifx\csname urlstyle\endcsname\relax
  \providecommand{\doi}[1]{doi: #1}\else
  \providecommand{\doi}{doi: \begingroup \urlstyle{rm}\Url}\fi

\bibitem[Casavola et~al.(2007)Casavola, Famularo, Franz{\`{e}}, and
  Sorbara]{Casavola2007}
Alessandro Casavola, Domenico Famularo, Giuseppe Franz{\`{e}}, and Michela
  Sorbara.
\newblock {Set-Points Reconfiguration in Networked Dynamical Systems}.
\newblock In \emph{Fault Detection, Supervision and Safety of Technical
  Processes 2006}, pages 132--137. Elsevier, 2007.
\newblock \doi{10.1016/b978-008044485-7/50023-3}.
\newblock URL \url{https://doi.org/10.1016/b978-008044485-7/50023-3}.

\bibitem[Machowski et~al.(2011)Machowski, Bialek, and
  Bumby]{machowski2011power}
J.~Machowski, J.W. Bialek, and J.~Bumby.
\newblock \emph{{Power System Dynamics: Stability and Control}}.
\newblock Wiley, 2011.
\newblock ISBN 9781119965053.
\newblock URL \url{https://books.google.de/books?id=wZv92UdKxi4C}.

\bibitem[Simpson-Porco et~al.(2016)Simpson-Porco, D\"{o}rfler, and
  Bullo]{SimpsonPorco2016}
John~W. Simpson-Porco, Florian D\"{o}rfler, and Francesco Bullo.
\newblock {Voltage collapse in complex power grids}.
\newblock \emph{Nature Communications}, 7\penalty0 (1), February 2016.
\newblock \doi{10.1038/ncomms10790}.
\newblock URL \url{https://doi.org/10.1038/ncomms10790}.

\bibitem[Ayasun et~al.(2004)Ayasun, Nwankpa, and Kwatny]{Ayasun2004}
S.~Ayasun, C.O. Nwankpa, and H.G. Kwatny.
\newblock {Computation of Singular and Singularity Induced Bifurcation Points
  of Differential-Algebraic Power System Model}.
\newblock \emph{{IEEE} Transactions on Circuits and Systems I: Regular Papers},
  51\penalty0 (8):\penalty0 1525--1538, August 2004.
\newblock \doi{10.1109/tcsi.2004.832741}.
\newblock URL \url{https://doi.org/10.1109/tcsi.2004.832741}.

\bibitem[Kwatny et~al.(1986)Kwatny, Pasrija, and Bahar]{Kwatny1986}
H.~Kwatny, A.~Pasrija, and L.~Bahar.
\newblock {Static bifurcations in electric power networks: Loss of steady-state
  stability and voltage collapse}.
\newblock \emph{{IEEE} Transactions on Circuits and Systems}, 33\penalty0
  (10):\penalty0 981--991, October 1986.
\newblock \doi{10.1109/tcs.1986.1085856}.
\newblock URL \url{https://doi.org/10.1109/tcs.1986.1085856}.

\bibitem[Sharafutdinov et~al.(2018)Sharafutdinov, Gorj{\~{a}}o, Matthiae,
  Faulwasser, and Witthaut]{Sharafutdinov2018}
Konstantin Sharafutdinov, Leonardo~Rydin Gorj{\~{a}}o, Moritz Matthiae, Timm
  Faulwasser, and Dirk Witthaut.
\newblock {Rotor-angle versus voltage instability in the third-order model for
  synchronous generators}.
\newblock \emph{Chaos: An Interdisciplinary Journal of Nonlinear Science},
  28\penalty0 (3):\penalty0 033117, March 2018.
\newblock \doi{10.1063/1.5002889}.
\newblock URL \url{https://doi.org/10.1063/1.5002889}.

\bibitem[Filatrella et~al.(2008)Filatrella, Nielsen, and
  Pedersen]{Filatrella2008}
G.~Filatrella, A.~H. Nielsen, and N.~F. Pedersen.
\newblock {Analysis of a power grid using a Kuramoto-like model}.
\newblock \emph{The European Physical Journal B}, 61\penalty0 (4):\penalty0
  485--491, February 2008.
\newblock \doi{10.1140/epjb/e2008-00098-8}.
\newblock URL \url{https://doi.org/10.1140/epjb/e2008-00098-8}.

\bibitem[Rohden et~al.(2012)Rohden, Sorge, Timme, and Witthaut]{Rohden2012}
Martin Rohden, Andreas Sorge, Marc Timme, and Dirk Witthaut.
\newblock {Self-Organized Synchronization in Decentralized Power Grids}.
\newblock \emph{Physical Review Letters}, 109\penalty0 (6), August 2012.
\newblock \doi{10.1103/physrevlett.109.064101}.
\newblock URL \url{https://doi.org/10.1103/physrevlett.109.064101}.

\bibitem[Schmietendorf et~al.(2017)Schmietendorf, Peinke, and
  Kamps]{Schmietendorf2017}
Katrin Schmietendorf, Joachim Peinke, and Oliver Kamps.
\newblock {The impact of turbulent renewable energy production on power grid
  stability and quality}.
\newblock \emph{The European Physical Journal B}, 90\penalty0 (11), November
  2017.
\newblock \doi{10.1140/epjb/e2017-80352-8}.
\newblock URL \url{https://doi.org/10.1140/epjb/e2017-80352-8}.

\bibitem[Sch\"{a}fer et~al.(2018{\natexlab{a}})Sch\"{a}fer, Beck, Aihara,
  Witthaut, and Timme]{Schfer2018}
Benjamin Sch\"{a}fer, Christian Beck, Kazuyuki Aihara, Dirk Witthaut, and Marc
  Timme.
\newblock {Non-Gaussian power grid frequency fluctuations characterized by
  L{\'{e}}vy-stable laws and superstatistics}.
\newblock \emph{Nature Energy}, 3\penalty0 (2):\penalty0 119--126, January
  2018{\natexlab{a}}.
\newblock \doi{10.1038/s41560-017-0058-z}.
\newblock URL \url{https://doi.org/10.1038/s41560-017-0058-z}.

\bibitem[Anvari et~al.(2016)Anvari, Lohmann, W\"{a}chter, Milan, Lorenz,
  Heinemann, Tabar, and Peinke]{Anvari2016}
M~Anvari, G~Lohmann, M~W\"{a}chter, P~Milan, E~Lorenz, D~Heinemann,
  M~Reza~Rahimi Tabar, and Joachim Peinke.
\newblock {Short term fluctuations of wind and solar power systems}.
\newblock \emph{New Journal of Physics}, 18\penalty0 (6):\penalty0 063027, June
  2016.
\newblock \doi{10.1088/1367-2630/18/6/063027}.
\newblock URL \url{https://doi.org/10.1088/1367-2630/18/6/063027}.

\bibitem[Zhang et~al.(2019)Zhang, Hallerberg, Matthiae, Witthaut, and
  Timme]{Zhang2019}
Xiaozhu Zhang, Sarah Hallerberg, Moritz Matthiae, Dirk Witthaut, and Marc
  Timme.
\newblock {Fluctuation-induced distributed resonances in oscillatory networks}.
\newblock \emph{Science Advances}, 5\penalty0 (7):\penalty0 eaav1027, July
  2019.
\newblock \doi{10.1126/sciadv.aav1027}.
\newblock URL \url{https://doi.org/10.1126/sciadv.aav1027}.

\bibitem[Yang et~al.(2017)Yang, Nishikawa, and Motter]{yang2017small}
Yang Yang, Takashi Nishikawa, and Adilson~E Motter.
\newblock {Small vulnerable sets determine large network cascades in power
  grids}.
\newblock \emph{Science}, 358\penalty0 (6365), 2017.

\bibitem[Sch\"{a}fer et~al.(2018{\natexlab{b}})Sch\"{a}fer, Witthaut, Timme,
  and Latora]{Schaefer2018}
Benjamin Sch\"{a}fer, Dirk Witthaut, Marc Timme, and Vito Latora.
\newblock {Dynamically induced cascading failures in power grids}.
\newblock \emph{Nature Communications}, 9\penalty0 (1), May 2018{\natexlab{b}}.
\newblock \doi{10.1038/s41467-018-04287-5}.
\newblock URL \url{https://doi.org/10.1038/s41467-018-04287-5}.

\bibitem[Wit(2021)]{Witthaut2021Nonlinear}
Collective nonlinear dynamics and self-organization in decentralized
  powergrids.
\newblock \emph{Rev. Mod. Phys., in revision}, 2021.

\bibitem[Owen(1997)]{Owen1997}
E.L. Owen.
\newblock {The origins of 60-Hz as a power frequency}.
\newblock \emph{{IEEE} Industry Applications Magazine}, 3\penalty0
  (6):\penalty0 8--14, November 1997.
\newblock \doi{10.1109/2943.628099}.
\newblock URL \url{https://doi.org/10.1109/2943.628099}.

\bibitem[Strogatz(2000{\natexlab{a}})]{Strogatz2000}
Steven~H. Strogatz.
\newblock {From Kuramoto to Crawford: exploring the onset of synchronization in
  populations of coupled oscillators}.
\newblock \emph{Physica D: Nonlinear Phenomena}, 143\penalty0 (1-4):\penalty0
  1--20, September 2000{\natexlab{a}}.
\newblock \doi{10.1016/s0167-2789(00)00094-4}.
\newblock URL \url{https://doi.org/10.1016/s0167-2789(00)00094-4}.

\bibitem[Strogatz(2000{\natexlab{b}})]{strogatz:2000}
Steven~H. Strogatz.
\newblock \emph{{Nonlinear Dynamics and Chaos: With Applications to Physics,
  Biology, Chemistry and Engineering}}.
\newblock Westview Press, 2000{\natexlab{b}}.

\bibitem[Manik et~al.(2014)Manik, Witthaut, Sch\"{a}fer, Matthiae, Sorge,
  Rohden, Katifori, and Timme]{Manik2014}
Debsankha Manik, Dirk Witthaut, Benjamin Sch\"{a}fer, Moritz Matthiae, Andreas
  Sorge, Martin Rohden, Eleni Katifori, and Marc Timme.
\newblock Supply networks: Instabilities without overload.
\newblock \emph{The European Physical Journal Special Topics}, 223\penalty0
  (12):\penalty0 2527--2547, September 2014.
\newblock \doi{10.1140/epjst/e2014-02274-y}.
\newblock URL \url{https://doi.org/10.1140/epjst/e2014-02274-y}.

\bibitem[Gerschgorin(1931)]{gerschgorin31}
S.~Gerschgorin.
\newblock \"{U}{ber die Abgrenzung der Eigenwerte einer Matrix}.
\newblock \emph{Izvestija Akademii Nauk SSSR, Serija Matematika}, 7\penalty0
  (3):\penalty0 749--754, 1931.

\bibitem[Stoer and Bulirsch(2002)]{stoer1989numerische}
Josef Stoer and Roland Bulirsch.
\newblock \emph{Numerische Mathematik}, volume~7.
\newblock Springer, 2002.

\bibitem[Timme et~al.(2004)Timme, Wolf, and Geisel]{timme2004topological}
Marc Timme, Fred Wolf, and Theo Geisel.
\newblock Topological speed limits to network synchronization.
\newblock \emph{Physical Review Letters}, 92\penalty0 (7):\penalty0 074101,
  2004.

\bibitem[Timme and Wolf(2008)]{timme2008simplest}
Marc Timme and Fred Wolf.
\newblock The simplest problem in the collective dynamics of neural networks:
  is synchrony stable?
\newblock \emph{Nonlinearity}, 21\penalty0 (7):\penalty0 1579, 2008.

\bibitem[Freund and Hoppe(2007)]{freund2007stoer}
Roland~W Freund and Ronald~W Hoppe.
\newblock \emph{Stoer/Bulirsch: Numerische Mathematik 1}.
\newblock Springer-Verlag, 2007.

\bibitem[Beardmore(2000)]{Beardmore2000}
R.~E. Beardmore.
\newblock {}double singularity-induced bifurcation points and singular hopf
  bifurcations.
\newblock 15\penalty0 (4):\penalty0 319--342, December 2000.
\newblock \doi{10.1080/713603759}.
\newblock URL \url{https://doi.org/10.1080/713603759}.

\bibitem[Menck et~al.(2014)Menck, Heitzig, Kurths, and Schellnhuber]{menck2014}
Peter~J. Menck, Jobst Heitzig, J\"{u}rgen Kurths, and Hans~Joachim
  Schellnhuber.
\newblock How dead ends undermine power grid stability.
\newblock 5\penalty0 (1), June 2014.
\newblock \doi{10.1038/ncomms4969}.
\newblock URL \url{https://doi.org/10.1038/ncomms4969}.

\bibitem[Van~Cutsem(2000)]{vanCutsem2000voltageInstability}
Thierry Van~Cutsem.
\newblock {Voltage instability: phenomena, countermeasures, and analysis
  methods}.
\newblock \emph{Proceedings of the IEEE}, 88\penalty0 (2):\penalty0 208--227,
  2000.

\bibitem[Sauer and Pai(1998)]{sauer1998power}
P.W. Sauer and M.A. Pai.
\newblock \emph{{Power System Dynamics and Stability}}.
\newblock Prentice Hall, 1998.
\newblock ISBN 9780136788300.
\newblock URL \url{https://books.google.de/books?id=dO0eAQAAIAAJ}.

\bibitem[Venkatasubramanian et~al.(1995)Venkatasubramanian, Schattler, and
  Zaborszky]{Schattler}
V.~Venkatasubramanian, H.~Schattler, and J.~Zaborszky.
\newblock {Local bifurcations and feasibility regions in differential-algebraic
  systems}.
\newblock \emph{IEEE Transactions on Automatic Control}, 40\penalty0
  (12):\penalty0 1992--2013, 1995.
\newblock \doi{10.1109/9.478226}.

\bibitem[Marszalek and Trzaska(2005)]{Marszalek}
W.~Marszalek and Z.W. Trzaska.
\newblock {Singularity-induced bifurcations in electrical power systems}.
\newblock \emph{IEEE Transactions on Power Systems}, 20\penalty0 (1):\penalty0
  312--320, 2005.
\newblock \doi{10.1109/TPWRS.2004.841244}.

\end{thebibliography}

\newpage 
\section{Appendix}

Here we motivate our statement in the main text about parameter configurations under which fixed points of Eq.\,(4a)-(4c), i.e., points where $\dot{\Theta}_i=\dot{\omega}_i= \dot{E}_i=0$ for all $i$, exist. 

\begin{theorem}
 The powers $P_i$ across the network of third order synchronous machines have to be in balance 
 \begin{align}
     0 = \sum_{i=1}^N P_i,  \label{eq:app_1}
 \end{align}
 to allow the entire system to settle to an equilibrium.
\end{theorem}
\begin{proof}
To show that such balance is a necessary condition for the existence of a fixed point, we take the sum over all $N$ nodes\cite{machowski2011power,Sharafutdinov2018}  of the rotor angle equation Eq.\,(4b), yielding
\begin{align}
    0 = \sum_{i=1}^N (P_i - \alpha_i \omega_i^* ) + \sum_{i=1}^N \sum_{j=1}^N B_{ij} E_i^* E_j^* \sin(\Theta^*_j - \Theta^*_i).
\end{align}
The sine functions are antisymmetric, while $B$ is symmetric against an exchange of indices such that the double sum equals zero. 
Furthermore, Eq.\,(4a) implies 
that $\omega_i^* = 0$ for all $i$ and therefore
\begin{align}
    0 = \sum_{i=1}^N P_i.
\end{align}
\end{proof}

The second assertion is that the voltage set points $E^f_i$ have to be sufficiently large for a fixed point to exist. We identify the effective coupling strengths $K_{ij}$ in Eq.\,(4b)
\begin{align}
    K_{ij} = B_{ij}E_i^* E_j^*.
\end{align}
From the paradigmatic Kuramoto model\cite{Strogatz2000} it is known that the coupling strength needs to be sufficiently large in order to compensate the parameters $P_i$ for all $i$ to allow the system to settle in a phase locked state. Due to the first condition, Eq.\,\eqref{eq:app_1} a phase locked state is also an equilibrium. 
\begin{theorem}
 The equilibrium coupling strength $K_{ij} = B_{ij} E_i^* E_j^*$ of the rotor angle dynamics is bound by 
 \begin{align}
     K_{ij} \leq B_{ij} (E^f + X_i \mu )^2 \leq B_{ij} (E^f)^2  \label{eq:app_5},
 \end{align}
 with $\mu \leq 0$ for networks of $N$ third order synchronous machines with $E^f_i = E^f$ for all $i \in \{1,2,...,N\}$.
\end{theorem}

\begin{proof}
We prove that relation Eq.\,\eqref{eq:app_5} holds for every synchronous machine individually. We assume an equilibrium of the entire system $\boldsymbol{\Theta}^* \in \mathbb{R}^N, \boldsymbol{E}^* \in \mathbb{R}^N$ exists. We exploit the following properties
\begin{align}
    E^f_i = E^f &> 0 \nonumber \\
    E_i^* &> 0 \nonumber \\
    X_i &> 0 \nonumber \\
    B_{ij} &\geq 0 \nonumber \\
    \Gamma_i &\leq 0 \nonumber \\
    \cos(x) &\leq 1 \text{ for all } x\in  \mathbb{R}
\end{align}
for all $i \in \{1,2,...,N\}$. Among the finite number of equilibrium voltage amplitudes, $E_i^*$ we pick the largest $E_i^*$ such that for all $j\neq i$
\begin{align}
    E_i^* \geq E_j^*
\end{align}
holds. For the synchronous machine $i$ the voltage amplitude equilibrium defining equations reads
\begin{align}
    0 =& E^f + (X_i \Gamma_i - 1) E_i^* \nonumber \\ &+ X_i \sum_{j=1,j\neq i}^N  B_{ij} (E_j^* \cos(\Theta_j^* - \Theta_i^*)  - E_i^*).
\end{align}
We exploit that $B_{ij}$, $E_j^*$, $X_i$ are nonnegativ, as well as the upper bound of $\cos(x)$ to evaluate 
\begin{align}
  0 \leq & E^f + (X_i \Gamma_i - 1) E_i^*  + X_i \sum_{j=1,j\neq i}^N  B_{ij} (E_j^*  - E_i^*) \nonumber \\
     = & E^f + (X_i \Gamma_i - 1)E_i^* + X_i \mu 
\end{align}
with $\mu \leq 0$ due to the fact that we have chosen the largest voltage amplitude $E_i^*$. We conclude 
\begin{align}
    0 \leq & E^f + (X_i \Gamma_i - 1) E_i^*  + X_i\mu  \nonumber \\
      \leq & E^f + (X_i \Gamma_i - 1) E_i^*   \nonumber \\
      \leq & E^f - E_i^* \label{eq:app_7}
\end{align}
by exploiting that $\Gamma_i X_i \leq 0$. The latter is equivalent to 
\begin{align}
    E_i^* \leq E^f, \label{eq:app_6}
\end{align}
and as $E_i^*\geq E_j^*$ for all $j \in \{1,2,...,N\}$ we have shown that relation Eq.\,\eqref{eq:app_6} holds for all $i \in \{1,2,...,N\}$.
The coupling strength $K_{ij} $ is bounded by

\begin{align}
    K_{ij} = B_{ij} E_i^* E_j^* \leq  B_{ij} (E^f)^2.
\end{align}
From this, we conclude that the parameter $E^f$ has to be set sufficiently large in order to provide sufficient coupling strength for the system to settle into an equilibrium. 
Furthermore, we show that $X_i$ has to be sufficiently small to guarantee the existence of an equilibrium. For this we start with Eq.\,\eqref{eq:app_7}
\begin{align}
    0 \leq & E^f + (X_i \Gamma_i - 1) E_i^*  + X_i\mu \nonumber \\
     \leq & E^f - E_i^*  + X_i\mu 
\end{align}
which is equivalent to 
\begin{align}
    E_i^* \leq E^f + X_i \mu
\end{align}
with negative $\mu$ as $E_i^*$ is again the largest voltage amplitude in the network. Increasing $X_i$ thus lowers the upper bound for the largest voltage amplitude and the upper bound for all coupling strengths $K_{ij}$ can be refined to
\begin{align}
 K_{ij} = B_{ij} E_i^* E_j^*  \leq  B_{ij} (E^f + X_i \mu)^2 
\end{align}
\end{proof}

\end{document}